\documentclass[11pt]{article}
\usepackage{amssymb,amsfonts,amsmath,amsthm}%
\usepackage[dvips]{graphicx,color}%
\newcommand{\figdraft}{false}%
\newcommand{\figfile}[1]{#1}%
\theoremstyle{plain}%
\newtheorem{theorem}{Theorem}[]%
\newtheorem{corollary}[theorem]{Corollary}%
\newtheorem{lemma}[theorem]{Lemma}%
\newtheorem{remark}[theorem]{Remark}%
\definecolor{colorGreen}{rgb}{0.,0.67,0}
\definecolor{colorRed}{rgb}{0.67,0.,0}
\definecolor{colorBlue}{rgb}{0.,0.,0.67}

\newcommand{\dblidx}[2]{_{#1,\,#2}}

\newcommand{\dblprm}[2]{{#1,\,#2}}

\newcommand{\iu}{\mathtt{i}}
\newcommand{\mhexp}[1]{{{\mathtt{e}}^{#1}}}

\newcommand{\sgn}{\mathrm{sgn}}

\newcommand{\fspace}[1]{{\mathsf{#1}}}
\newcommand{\fspaceL}{\fspace{L}}
\newcommand{\fspaceC}{\fspace{C}}

\newcommand{\ol}[1]{{\overline{#1}}}

\newcommand{\phase}{{\varphi}}

\newcommand{\Rset}{{\mathbb{R}}}
\newcommand{\Zset}{{\mathbb{Z}}}

\newcommand{\Nset}{{\mathbb{N}}}

\newcommand{\cointerval}[2]{[#1,\,#2)}%
\newcommand{\oointerval}[2]{(#1,\,#2)}%
\newcommand{\ccinterval}[2]{[#1,\,#2]}%

\newcommand{\DO}[1]{{O\at{#1}}}

\newcommand{\loc}{{\rm loc}}

\newcommand{\tdots}{{...}}%

\newlength{\mhpicDwidth}
\newlength{\mhpicDvsep}
\newlength{\mhpicDhsep}

\newlength{\mhpicPwidth}
\newlength{\mhpicPvsep}
\newlength{\mhpicPhsep}

\newlength{\mhpicWhsep}

\newcommand{\at}[1]{{\left({#1}\right)}}
\newcommand{\nat}[1]{(#1)}
\newcommand{\bat}[1]{{\big(#1\big)}}

\newcommand{\ul}[1]{\underline{#1}}

\newcommand{\bigpar}{\par\quad\newline\noindent}

\newcommand{\abs}[1]{\left|{#1}\right|}

\newcommand{\dint}[1]{\,\mathrm{d}#1}

\newcommand{\eps}{{\varepsilon}}

\newcommand{\ka}{{\kappa}}
\newcommand{\la}{{\lambda}}
\newcommand{\si}{{\sigma}}

\newcommand{\calD}{\mathcal{D}}
\newcommand{\calE}{\mathcal{E}}
\newcommand{\calF}{\mathcal{F}}
\newcommand{\calG}{\mathcal{G}}
\newcommand{\calH}{\mathcal{H}}
\newcommand{\calI}{\mathcal{I}}

\newcommand{\calM}{\mathcal{M}}
\newcommand{\calN}{\mathcal{N}}

\usepackage[%
a4paper,%
nohead,%
twoside,
ignorefoot,
scale=0.8,
bindingoffset=.4cm
]{geometry}%

\usepackage{float}%
\begin{document}%
%
%
\title{Heteroclinic  standing waves in defocussing DNLS equations%
\\\emph{Variational approach via energy minimization }}%
\date{\today}%
\author{%
Michael Herrmann%
\thanks{
    Oxford Centre for Nonlinear PDE (OxPDE),
    {\tt{michael.herrmann@maths.ox.ac.uk}}
}%
}%
\maketitle
%
%
%
\begin{abstract}%
We study heteroclinic standing waves (dark solitons) in discrete nonlinear Schr\"{o}dinger
equations with defocussing nonlinearity. Our main result is a quite elementary existence
proof for waves with monotone and odd profile, and relies on minimizing an appropriately
defined energy functional. We also study the continuum limit and the numerical approximation
of standing waves.
\end{abstract}%
%
%

\quad\newline\noindent%
\begin{minipage}[t]{0.15\textwidth}%
Keywords: %
\end{minipage}%
\begin{minipage}[t]{0.8\textwidth}%
\emph{discrete nonlinear Schr\"{o}dinger equation (DNLS)}, \\%
\emph{nonlinear lattice waves, variational methods}, \\%
\emph{heteroclinic standing waves, dark solitons} %
\end{minipage}%
\medskip
\newline\noindent
\begin{minipage}[t]{0.15\textwidth}%
MSC (2000): %
\end{minipage}%
\begin{minipage}[t]{0.8\textwidth}%
37K60, 
47J30, 
78A40  

\end{minipage}%
%
%
%
%
%
%
%

\section{Introduction}
%
Discrete nonlinear Schr\"{o}dinger equations (DNLS) are prototypical models for the dynamics
of coupled oscillators and have a broad range of application in nonlinear optics and particle
physics, see \cite{KRB01,EJ03,Kev:DNLS:Ch1,Por09}. A special focus in the mathematical
analysis of such systems lies on the investigation of coherent structures (travelling or
standing waves, breathers, kinks) as these describe the fundamental modes for energy
localization and wave propagation in nonlinear discrete media.
\par
This paper concerns heteroclinic standing waves for the one-dimensional DNLS equation
\begin{align}%
\label{Eqn:DNLS}%
\iu\dot{A}_j-\beta\at{A_{j+1}+A_{j-1}-2A_j}+\Psi^\prime\nat{\abs{A_j}^2}A_j=0.
\end{align}
Here $A_j=A_j\at{t}$ denotes the complex amplitude of the oscillator $j\in{Z}$ at time $t$
and $\beta$ is the coupling constant. The potential function $\Psi$ describes the energy of a
single oscillator system and is often  assumed to be monomial. In what follows we consider an
infinite number of oscillators and set either $Z=\Zset$ (\emph{on-site setting}) or
$Z=\tfrac{1}{2}+\Zset$ (\emph{inter-site setting}). The difference between both settings will
be explained below.
\par
The lattice system \eqref{Eqn:DNLS} is an infinite dimensional system of Hamiltonian ODEs
(well-posedness results are given in \cite{PKP05,GP10}) and possesses the conserved
quantities
\begin{align}
\label{Eqn:ConsLaws}%
\calH\at{A}=\sum\limits_{j}\Psi\bat{\abs{A}_j^2}+\beta\sum\limits_{j}\abs{A_{j+1}-A_j}^2
,\qquad
\calN\at{A}=\sum\limits_j{}\abs{A_j}^2.
\end{align}
More precisely, $\calH$ is the Hamiltonian of \eqref{Eqn:DNLS}, whereas $\calN$ gives the
power of a state $A$ and stems from the gauge symmetry, that is the invariance of
\eqref{Eqn:DNLS} under the transformation $A_j\at{t}\mapsto\mhexp{\iu\phase}A_j\at{t}$ with
$\phase\in\Rset$.
\bigpar
Standing waves are exact solutions to \eqref{Eqn:DNLS} that satisfy
$A_j\at{t}=\mhexp{\iu\sigma{}t}u_j\at{t}$ with \emph{frequency} $\sigma\in\Rset$ and
real-valued \emph{profile} $u=\at{u_j}_{j\in{Z}}\in\ell^\infty\at{Z}$. Standing waves can be
regarded as relative equilibria with respect to the gauge symmetry and satisfy
\begin{align}
\label{Eqn:StandingWave0}
\sigma{u}_j=-\beta\at{u_{j+1}+u_{j-1}-2{u_j}}+\Psi^\prime\nat{u_j^2}u_j.
\end{align}
Heteroclinic waves connect different asymptotic states via
\begin{align}
\label{Eqn:AsmptoticStates}
u_j\xrightarrow{j\to\pm\infty}{u_{\pm\infty}}
\end{align}
and correspond to so called \emph{dark solitons}. Other types of standing waves are periodic
waves with $u_j=u_{j+N}$ for some $N$, and \emph{bright solitons}, which are homoclinic with
$\lim_{j\to\pm\infty}u_j=0$.
\par
In this paper we aim in establishing the existence of standing
waves in the set
\begin{align*}
\calM=\{%
u\in\ell^\infty\at{Z}\;:\;u_{-j}=-u_j,\;-u_\infty\leq{}u_{j}\leq{u_{j+1}}\leq{u_\infty}
\;\;\;\forall\;j\in{Z}\},
\end{align*}
which consists of all profiles that are odd, increasing, and take values in
$\ccinterval{-u_\infty}{u_\infty}$. In particular, we suppose $0<u_\infty=-u_{-\infty}$.
Notice that we have $u_0=0$ for all on-site waves, whereas inter-site waves generically
satisfy $u_{-1/2}<0<u_{1/2}$.%
\bigpar
For convex $\Psi$ one distinguishes between the \emph{focussing} and the \emph{defocussing}
case, which (in our notation) correspond to $\beta<0$ and $\beta>0$, respectively. Although
both cases are linked via the \emph{staggering transformation}
$u_j\rightsquigarrow\at{-1}^ju_j$ they describe different physical situations. In the context
of standing waves it is well established, see \cite{Kev:DNLS:Ch2} and \cite{Kev:DNLS:Ch5},
that the most fundamental (that means most stable) waves are homoclinic for $\beta<0$ but
heteroclinic for $\beta>0$, respectively. In what follows we solely consider the defocussing
case $\beta>0$ and construct heteroclinic \emph{single-pulse} waves by \emph{minimizing} an
energy functional. The analogues in the focussing case are homoclinic single-pulse waves,
which can be constructed by \emph{constrained maximization}, see \cite{Her10PreC}.
\par
Several methods have been used to prove the existence of standing waves in DNLS equations.
\emph{Continuation methods} were introduced by MacKay and Aubry \cite{MA94,Aub97} and have
been proven powerful for both theoretical considerations and numerical computations, see for
instance \cite{KL09}. Continuation method rely on the observation that
\eqref{Eqn:StandingWave0} can be solved explicitly in both the anti-continuum limit
$\beta\to0$ and the continuum limit $\beta\to\infty$. During the last years, however, there
has been a growing interest in other approaches to the existence problem for standing waves.
We refer to \cite{PR05}, which exploits spatial dynamics and centre manifold reduction, and
to the variational methods in \cite{PZ01,PR08,ZP09,ZL09}, which rely on critical point
techniques (linking theorems, Nehari manifold).
\bigpar
We now summarize the main idea in our variational existence proof. At first we notice that
\eqref{Eqn:StandingWave0} and \eqref{Eqn:AsmptoticStates} couple the frequency $\si$ to the
asymptotic states via
\begin{align*}
\si=\Psi^\prime\at{u^2_{\infty}}.
\end{align*}
Secondly, we introduce the function
\begin{align*}
F\at{\eta}=\Psi\at{\eta^2}-\Psi\at{u_\infty^2}-\Psi^\prime\at{u_\infty^2}\at{\eta^2-u_\infty^2},
\end{align*}
and rewrite the standing wave equation as
\begin{align*}
F^\prime\at{u_j}-2\beta\at{u_{j+1}+u_{j-1}-2u_j}=0.
\end{align*}
The key observation is that each standing wave is a critical point of the energy functional
$\calE$ with
\begin{align}
\label{Eqn:StandingWave1}
\calE\at{u}=\calF\at{u}+\beta\calD\at{u}
,\qquad
\calF\at{u}=\sum\limits_jF\at{u_j}
,\qquad
\calD\at{u}=\sum\limits_j\at{u_{j+1}-u_j}^2.
\end{align}
The energy $\calE$ is naturally related to the conserved quantities from \eqref{Eqn:ConsLaws}.
In fact, on a formal level we find
\begin{align*}
\calE\at{u}=\calH\at{u}-\sigma\calN\at{u}-
\calH\at{\bar{u}}+\si\calN\at{\bar{u}}+\beta\calD\at{\bar{u}}
\end{align*}
where $\bar{u}$ is an arbitrary reference profile with $\bar{u}_j^2=u_\infty^2$ for all $j$,
but notice that the $\calN$- and $\calH$-terms are infinite if $u$ satisfies
\eqref{Eqn:AsmptoticStates}. $\calE\at{u}$, however, is well defined as long as $u$
approaches the asymptotic states sufficiently fast. In fact, using
$F\at{\pm{u_\infty}}=F^\prime\at{\pm{u_\infty}}=0$ we find that both $\calF\at{u}$ and
$\calD\at{u}$ are finite if the sequence $j\mapsto{u_j}-u_\infty\sgn{j}$ belongs to
$\ell^2\at{Z}$.
\par
\par
Our strategy for proving the existence of standing waves is to show that $\calE$ attains its
minimum on $\calM$ by using the direct method from the calculus of variations. Afterwards we
show that each minimizer satisfies \eqref{Eqn:StandingWave0} as it is strictly increasing.
Our existence result for energy minimizing waves can be summarized as follows.
\begin{theorem}
\label{Intro:Result1}%
Let $u_\infty>0$ be given and let $\Psi$ be twice continuously differentiable on
$\ccinterval{0}{x_\infty}$ with $x_\infty=u_\infty^2$. Moreover, suppose that
\begin{align}
\label{Intro:Result1.Eqn1}
\quad\Psi^{\prime\prime}\at{x_\infty}>0
\quad\text{and}\quad
\quad\Psi\at{x}>\Psi\at{x}+\Psi^\prime\at{x_\infty}\at{x-x_\infty}
\quad\text{for all}\quad0\leq{x}<x_\infty.
\end{align}
Then, for each $\beta>0$ the functional $\calE$ attains its minimum on $\calM$. Each
minimizer $u\in\calM$ is strictly increasing, converges exponentially to ${\pm}u_{\infty}$ as
$j\to\pm\infty$, and solves the standing wave equation \eqref{Eqn:StandingWave0} with
frequency $\sigma=\Psi^\prime\at{x_\infty}$.
\end{theorem}
We proceed with some remarks concerning the assumptions and assertions of Theorem
\ref{Intro:Result1}.
\begin{enumerate}
\item
Theorem \ref{Intro:Result1} holds in both the on-site and inter-site setting. Moreover,
via $u_j\rightsquigarrow-u_{j}$ it also provides also the existence of standing waves
with decreasing profile.
\item
Suppose that $\Psi$ is strictly convex on $\cointerval{0}{\infty}$ with
$\Psi^\prime\at{0}=0$. Then we have $\Psi\at{x}-\Psi\at{y}\geq\Psi^\prime\at{y}\at{x-y}$
for all $x,y\geq0$ and Theorem \ref{Intro:Result1} guarantees the existence of standing
waves with arbitrary $u_\infty$. Furthermore, there exist standing waves for non-convex
$\Psi$, see the examples in \S\ref{sec:num}.
\item
Assumption \eqref{Intro:Result1.Eqn1}$_2$ precisely means that $F$ is positive on the
interval $\oointerval{-u_\infty}{u_\infty}$. This condition implies
$F^{\prime\prime}\at{{u_\infty}}\geq0$ and is sharp in the following sense. Suppose there
exists $0\leq{u_\ast}\leq{u_\infty}$ such that $F\at{u_\ast}<0$. Then the energy $\calE$
is unbounded from below and global minimizers of $\calE$ can therefore not exist. See
also the discussion in \S\ref{sec:num}.
\item
Condition \eqref{Intro:Result1.Eqn1}$_1$ is equivalent to
$F^{\prime\prime}\at{{u_\infty}}>0$ and guarantees that each standing wave approaches its
asymptotic states exponentially. For $F^{\prime\prime}\at{{u_\infty}}=0$ one can still
prove the existence of energy minimizing waves (using some limit procedure), but since
$\calE$ has no minimizer for $F^{\prime\prime}\at{{u_\infty}}<0$ these waves are
non-generic and have an algebraic tail. 
\item
The standing waves from Theorem \ref{Intro:Result1} are so called single-pulse dark
solitons. From continuation results we know that there also exists an infinite number of
multi-pulse dark solitons corresponding to non-monotone solutions to
\eqref{Eqn:StandingWave0} and \eqref{Eqn:AsmptoticStates}, see for instance
\cite{Kev:DNLS:Ch5}. However, multi-pulse waves are expected to be unstable and we
conjecture that they correspond to genuine saddle points of $\calE$. A detailed
investigation of this issue is left for future research.
\end{enumerate}
The proof of Theorem \ref{Intro:Result1} is given in \S\ref{sec:Waves.1} and
\S\ref{sec:Waves.2}. In \S\ref{sec:Waves.3} we establish an approximation result which allows
to compute standing waves on finite index sets and provides the base for the numerical
simulations in \S\ref{sec:num}. Moreover, in \S\ref{sec:Waves.Limit} we study the continuum
limit of standing waves with fixed asymptotic states. To this end we introduce a scaling
small parameter $\eps>0$ and scale \eqref{Eqn:StandingWave1} by $j\rightsquigarrow\eps{j}$
and $\beta\rightsquigarrow\beta/\eps^2$. Then we proof that the standing waves converge as
$\eps\to0$ to a heteroclinic solution of
\begin{align}
\label{Eqn:LimitProblem}
2\beta{u}^{\prime\prime}\at\xi=-F^\prime\bat{u\at{\xi}}
,\quad\lim_{j\to\pm\infty}u\at\xi=\pm{u_\infty}.
\end{align}
%
%
%
\section{Existence and properties of standing waves }\label{sec:Waves}
%
In order to prove Theorem \ref{Intro:Result1} we first observe that \eqref{Eqn:StandingWave0}
is invariant under the scaling
\begin{align*}
u_j\mapsto\eta{u_j}
,\qquad%
\sigma\mapsto{\tau}\sigma
,\qquad%
\beta\mapsto{\tau}\beta
,\qquad%
\Psi\at{x}\mapsto{\tau}{\eta}^{-2}\Psi\at{\eta^2x}+\delta
\end{align*}
with arbitrary $\eta,\,\tau,\,\delta\in\Rset$. From now on we therefore assume that
\begin{align*}
\si=\Psi\at{1}=\Psi^\prime\at{1}=1,\qquad
u_{\pm\infty}=\pm1,\qquad
F\at{\eta}=\Psi\at{\eta^2}-\eta^2,
\end{align*}
and study the \emph{normalized} standing wave equation
\begin{align} \label{Eqn:StandingWave}
F^\prime\at{u_j}-2\beta\at{u_{j+1}+u_{j-1}-2u_j}=0
,\qquad\lim_{j\to\pm\infty}u_j=\pm1.
\end{align}
%
%
%
%
%
\subsection{Standing waves as minimizers of $\calE$}\label{sec:Waves.1}
%
Our assumptions on $\Psi$ ensure that $\calE$ has nice properties on $\calM$. Recall that
$\calM$ is compact with respect to the weak$\star$ topology in $\ell^\infty\at{Z}$.
\begin{lemma}
\label{Rem:MProps}%
$\calE$ is non-negative and weakly$\star$ lower semi-continuous on $\calM$. Moreover,
$\calE\at{u}<\infty$ implies $\lim\limits_{j\to\pm\infty}u_j=\pm{1}$ and
\begin{align}
\label{Rem:MProps.Eqn1}%
\lim_{t\to0}\tfrac{1}{t}\at{\calE\at{u+t{v}}-\calE\at{u}}
=\sum_{j\in{Z}}v_j\,\calG\at{u}_j
\end{align}
for all $v\in\ell^1\at{Z}$, where $\calG\at{u}\in\ell^\infty\at{Z}$ with
\begin{align*}
\calG\at{u}_j=\Psi^\prime\at{u_j^2}u_j-u_j-\beta\at{u_{j+1}+u_{j-1}-2u_j}
\end{align*}
is the G\^{a}teaux derivative of $\calE$ in $u$.
\end{lemma}
\begin{proof}
Let $u\in\calM$ be given. Then we have $\calE\at{u}\geq\calF\at{u}\geq0$ due to the
assumption on $F$, and the monotonicity of $u$ implies
$\lim_{j\to\pm\infty}u_j=\pm{}u_\infty$ for some $u_\infty\in\ccinterval{0}{1}$. Assuming
$u_\infty<1$ we find $F\at{u_j}\geq\tfrac{1}{2}F\at{u_\infty}>0$ for almost all $j$, and
hence $\calE\at{u}\geq\calF\at{u}=\infty$. Moreover, \eqref{Rem:MProps.Eqn1} follows from a
direct computation. Finally, we consider a sequence $\at{u_n}_{n\in\Nset}\subset\calM$, with
$u_n=\at{u_\dblprm{n}{j}}_{j\in{Z}}$ for all $n$, that converges weakly$\star$ in
$\ell^\infty\at{Z}$ to some limit $u\in\calM$. Then we have $u\dblidx{n}{j}\to{u_j}$ for all
$j$, and Fatou's Lemma gives $\calE\at{u}\leq\liminf_{n\to\infty}\calE\at{u_n}$.
\end{proof}
We are now able to prove our main result.
\begin{theorem}
\label{Theo:ExistenceOfWaves}%
$\calE$ attains its minimum on $\calM$. Moreover, each minimizer is strictly increasing and
satisfies the standing wave equation \eqref{Eqn:StandingWave}.
\end{theorem}
\begin{proof}
The existence of minimizers follows from the weak$\star$ compactness of $\calM$ and the lower
semi-continuity of $\calE$. We now show by contradiction that each minimizer $u\in\calM$ is
strictly increasing. Suppose at first that there exists some $j_0\in{Z}$ such that
\begin{align*}
u_{j_0-1}<1\quad\text{and}\quad{u_j=1}
\quad\text{for all}\quad%
{j\in{Z}}\quad\text{with}\quad{j\geq{j_0}},
\end{align*}
and define $v\in\ell^1\at{Z}$ by $v_{j_0}=-1$ and $v_j=0$ for all $j\neq{j_0}$. Then we have
\begin{align*}
\calG\at{u}_{j_0}=\Psi^\prime\at{1}1-1-\beta\at{1+u_{j_0-1}-2}=\beta\at{1-u_{j_0-1}}>0.
\end{align*}
and \eqref{Rem:MProps.Eqn1} yields $\calE\at{u+tv}<\calE\at{u}$ for all sufficiently small
$t>0$, which is the desired contradiction. Consequently, we have $u_j<1$ for all $j\in{Z}$
and $-1<u_j$ follows similarly. Secondly, suppose there exist a constant plateau between $-1$
and $+1$, that means
\begin{align*}
u_{j_1-1}<u_{j_1}=u_{j_1+1}=\tdots=u_{j_2-1}=u_{j_2}<u_{j_2+1}
\end{align*}
for some indices $j_1<j_2$, and define $v\in\ell^1\at{Z}$ by $v_{j_1}=-1$, $v_{j_2}=+1$, and
$v_j=0$ for all $j\in{Z}\setminus\{j_1,\,j_2\}$.  This implies
\begin{align*}
\calG\at{u}_{j_1}=\Psi^\prime\at{u_{j_1}^2}u_{j_1}-u_{j_1}+\beta\at{u_{j_1}-u_{j_1-1}}>
\Psi^\prime\at{u_{j_2}^2}u_{j_2}-u_{j_2}+\beta\at{u_{j_2}-u_{j_2+1}}=\calG\at{u}_{j_2},
\end{align*}
and \eqref{Rem:MProps.Eqn1} gives $\calE\at{u+tv}<\calE\at{u}$, which is again a
contradiction. Finally, since $u$ is strictly increasing we can perturb the $u_j$'s
independently from each other, so \eqref{Rem:MProps.Eqn1} yields $\calG\at{u}_j=0$ for all
$j\in{Z}$.
\end{proof}
Recall that Theorem \ref{Theo:ExistenceOfWaves} provides the existence of standing waves in
both the on-site and off-site setting.
%
%
\subsection{Exponential tails for standing waves}\label{sec:Waves.2}
%
%
We next show that standing waves converge exponentially to the asymptotic states.
Heuristically, the decay rate $\la$ is determined by linearizing \eqref{Eqn:StandingWave} in
the asymptotic states, i.e., $\la$ is the unique positive solution to
\begin{align}
\label{Eqn:DecayRate}
4\beta\at{\cosh{\la}-1}=F^{\prime\prime}\at{1}=4\Psi^{\prime\prime}\at{1}.
\end{align}
From this we conclude that $\Psi^{\prime\prime}\at{1}>0$ is truly necessary for standing
waves to have exponential tails.
\begin{lemma}
%
Let $u\in\calM$ be any strict monotone solution to \eqref{Eqn:StandingWave} and choose
$\ul\la$, $\ol{\la}$ such that $0<\ul\la<\la<\ol{\la}$ with $\la$ as in
\eqref{Eqn:DecayRate}. Then there exist constants $\ul{c}$ and $\ol{c}$ such that
\begin{align*}
{\ul{c}}\exp\at{-\ol{\la}\abs{j}}\leq\abs{\sgn{j}-u_j}\leq{\ol{c}}\exp\at{-\ul{\la}\abs{j}}
\end{align*}
for all $j\in{Z}$.
\end{lemma}
\begin{proof}
By virtue of $u_{-j}=-u_{j}$ it is sufficient to consider positive $j$. For $1<j\in{Z}$ we
define $w_j$ and $\kappa_j$ by $w_j=1-u_j$ and $\ka_j={w_j}/w_{j-1}$. By $u\in\calM$ and
\eqref{Eqn:StandingWave} we have
\begin{align*}
0<{w_j}<{w_{j-1}}<1,\quad%
\delta_jw_j=w_{j+1}+w_{j-1}-2w_j,\quad\delta_j=\frac{F^{\prime\prime}\at{\xi_j}}{2\beta},
\end{align*}
where $\xi_j$ denotes some intermediate value in $\ccinterval{u_j}{1}$, and hence
\begin{align*}
0<{\ka_j}<1
,\quad%
\kappa_{j+1}=2+\delta_j-\frac{1}{\kappa_j}.
\end{align*}
Moreover, using $w_j\to0$ as $j\to\infty$ we find
\begin{align*}
\delta_j\xrightarrow{j\to\infty}\delta_\infty=\frac{F^{\prime\prime}\at{1}}{2\beta},
\quad\kappa_j\xrightarrow{j\to\infty}\kappa_\infty
=%
\frac{2+\delta_\infty-\sqrt{\delta_\infty}\sqrt{4+\delta_\infty}}{2},
\end{align*}
and a direct computation reveals that $\la=-\ln\at{\kappa_\infty}>0$ satisfies
\eqref{Eqn:DecayRate} and $\ul\la<\la<\ol\la$. The desired result now follows immediately.
\end{proof}
We mention that \eqref{Eqn:DecayRate} implies the expansions
\begin{align*}
\la=-\ln\at{\delta^{-1}}+2\delta^{-1}\at{1+\DO{\delta^{-1}}}
,\quad%
\la=\sqrt{\delta}\at{1+\DO{\delta}}
,\quad%
\delta=\frac{F^{\prime\prime}\at{1}}{2\beta}.
\end{align*}
In particular, we have $\la\to\infty$ in the anti-continuum limit $\beta\to0$ but $\la\to0$
in the continuum limit $\beta\to\infty$.
%
%
\subsection{Ritz approximation of standing waves}\label{sec:Waves.3}
%
%
It is reasonable, and useful for numerical simulations, to approximate the standing waves
from Theorem \ref{Theo:ExistenceOfWaves} by minimizing $\calE$ on the finite-dimensional set
\begin{align}
\label{Eqn:RitzSet}
\calM_N=\{u\in\calM\;:\;u_j=\sgn{j}
\;\;\;\forall\;j\in{Z}\;\text{with}\;\abs{j}>N\},\quad{N}\in\Nset.
\end{align}
Notice that $\calE$ attains its minimum on $\calM_N$ as $\calM_N$ is closed under weak$\star$
convergence in $\ell^\infty\at{Z}$.
\begin{lemma}
%
We have
\begin{align}
\label{Lem:ConvergenceMinima.Eqn1}
\min\calE|_{\calM_N}\xrightarrow{N\to\infty}\min\calE|_{\calM}.
\end{align}
Moreover, each sequence $\at{u_N}_{N\in\Nset}$ with $u_N\in\calM_N$ and
$\calE\at{u_N}=\min\calE|_{\calM_N}$ for all $N$, has a subsequence that converges pointwise
to some $u$ with $\calE\at{u}=\min\calE|_{\calM}$.
\end{lemma}
\begin{proof}
Let $u\in\calM$ be a minimizer of $\calE$, and define the sequence $\at{u_N}_N$ by
$u\dblidx{N}{j}=u_j$ for $\abs{j}\leq{N}$ and $u\dblidx{N}{j}=\sgn{j}$ for $\abs{j}>N$. Then
we have
\begin{align*}
\calF\at{u_N}=\sum\limits_{\abs{j}\leq{N}}F\at{u_j}\leq\calF\at{u}
\end{align*}
and
\begin{align*}
\calD\at{u_N}
=%
\sum\limits_{\abs{j}
\leq%
{N}-1}\at{u_{j+1}-u_{j}}^2+2\at{1-u_N}^2\leq\calD\at{u}+2\at{1-u_N}^2.
\end{align*}
Taking the limit $N\to\infty$ gives
\begin{align*}
\lim_{N\to\infty}\min\calE|_{\calM_N}
\leq%
\lim_{N\to\infty}\calE\at{u_N}\leq\calE\at{u}=\min\calE|_{\calM}.
\end{align*}
This implies \eqref{Lem:ConvergenceMinima.Eqn1} because we have
$\calM_{N_1}{\subset}\calM_{N_2}{\subset}\calM$ and hence
$\min\calE|_{\calM}\leq\min\calE|_{\calM_{N_2}}\leq\min\calE|_{\calM_{N_1}}$ for all
$N_1<N_2$. Finally, suppose that $\at{u_N}_N$ satisfies $\calE\at{u_N}=\min\calE|_{\calM_N}$.
By weak$\star$ compactness we can extract a (not relabelled) subsequence that converges
weakly$\star$ in $\ell^\infty\at{Z}$ to some limit $u\in\calM$.  Lemma \ref{Rem:MProps}
combined with \eqref{Lem:ConvergenceMinima.Eqn1} gives
\begin{align*}
\min\calE|_{\calM}
\leq%
\calE\at{u}
\leq%
\lim\limits_{N\to\infty}\calE\at{u_{N}}
=\lim\limits_{N\to\infty}\min\calE|_{\calM_{N}}=\min\calE|_{\calM},
\end{align*}
so $u$ is in fact a minimizer of $\calE$.
\end{proof}
%
%
%
%
\subsection{Continuum limit of standing waves}\label{sec:Waves.Limit}
%
%
In this section we characterize the continuum limit of discrete standing waves. To this end
we consider profile functions $u=u\at{\xi}\in\calM$, where $\xi\in\Rset$ is continuous
variable and $\calM$ is given by
\begin{align*}
\calM=\left\{
u\in\fspaceL^\infty\at{\Rset}\;:\;-1\leq{u\at{\xi_1}}\leq{u\at{\xi_2}}=-u\at{-\xi_2}\leq1
\;\;\;\forall\;\xi_1\leq\xi_2\in\Rset
\right\}.
\end{align*}
The set $M$ is a naturally related to the formal limit problem \eqref{Eqn:LimitProblem} as
standard arguments for planar Hamiltonian ODEs imply the following result.
\begin{remark}
\label{Rem:Cont.LimitSol}
There exists a unique $u\in\calM$ with $F^\prime\at{u}=2\beta{u}^{\prime\prime}$.
\end{remark}

We now introduce a small lattice parameter $\eps>0$ and consider the functional $\calE_\eps$
in $\calM$ defined by
\begin{align*}
\calE_\eps\at{u}=\calF\at{u}+\beta\calD_\eps\at{u}
,\quad
\calF\at{u}=\int\limits_{\Rset}F\at{u\at{\xi}}\dint\xi
,\quad
\calD_\eps\at{u}=\int\limits_{\Rset}\at{\frac{u\at{\xi+\eps}-u\at{\xi}}{\eps}}^2\dint\xi.
\end{align*}
The connection with the discrete setting becomes obvious when introducing
\begin{align*}
\calM_\eps=\left\{
u\in\calM\;:\;u\at{\eps{j}}=u\at{\eps{j}+\eps\xi}
\;\;\;\forall\;j\in{Z},\;\xi\in\oointerval{-\tfrac{1}{2}}{\tfrac{1}{2}}\in\Rset
\right\}.
\end{align*}
Each $u\in\calM_\eps$ is piecewise constant and can be identified with a increasing and odd
profile from $\ell^\infty\at{Z}$. In particular, we have
\begin{align*}
\calF\at{u}=\eps\sum\limits_{j\in{Z}}F\at{u\at{\eps{j}}}
,\quad%
\calD_\eps\at{u}
=%
\frac{1}{\eps}\sum\limits_{j\in{Z}}\at{u\at{\eps{j}+\eps}-u\at{\eps{j}}}^2.
\end{align*}
The result from \S\ref{sec:Waves.1} imply that $\calE_\eps$ attains its minimum on
$\calM_\eps$.
\begin{lemma}
\label{Lem:Cont.ExistenceOfWaves}
For each $\eps>0$ there exist a minimizer $u_\eps\in\calM_\eps$ for
$\calE_\eps|_{\calM_\eps}$ such that
\begin{align}
\label{Lem:Cont.ExistenceOfWaves.Eqn1}
\int\limits_{\Rset}F^\prime\at{u_\eps\at\xi}\varphi\at\xi\dint\xi
=%
2\beta\int\limits_{\Rset}
u_\eps\at\xi\frac{\varphi\at{\xi+\eps}+\varphi\at{\xi+\eps}-2\varphi\at{\xi}}{\eps^2}
\dint\xi
\end{align}
holds for all test functions $\varphi\in\fspaceC^\infty_0\at{\Rset}$. Moreover, there exists
a constant $C$ such that $\calE_\eps\at{u_\eps}\leq{C}$ for all $0<\eps\leq1$.
\end{lemma}
\begin{proof}
The existence of $u_\eps$ follows from Theorem \ref{Theo:ExistenceOfWaves} (with
$\beta/\eps^2$ instead of $\beta$), and \eqref{Lem:Cont.ExistenceOfWaves.Eqn1} is the weak
formulation of the discrete standing wave equation, see \eqref{Eqn:StandingWave}. We now
consider $v_\eps\in\calM_\eps$ with
\begin{align*}
v_\eps\at{\eps{j}}=\left\{
\begin{array}{ll}
-1&\text{for all $j\in{Z}$ with $\eps{j}\leq-1$},
\\%
+1&\text{for all $j\in{Z}$ with $\eps{j}\geq+1$},
\\%
\eps{j}&\text{for all $j\in{Z}$ with $\eps\abs{j}<1$}.
\end{array}
\right.
\end{align*}
A direct computation shows
\begin{align*}
\calF\at{v_\eps}=\sum\limits_{\eps\abs{j}<1}\eps{}F\at{\eps{j}}
=\int\limits_{-1}^{1}{F}\at{\eta}\dint\eta+\DO{\eps},\quad
\calD_\eps\at{v_\eps}=\sum\limits_{\eps\abs{j}\leq{1}}\eps={2}+\DO{\eps}.
\end{align*}
The uniform bounds for $\calE_\eps\at{u_\eps}$ now follows due to
$\calE_\eps\at{u_\eps}\leq\calE_\eps\at{v_\eps}$.
\end{proof}
We are now able to pass to the continuum limit $\eps\to0$.
\begin{corollary}
We have $u_\eps\to{u}$ as $\eps\to0$ in $\fspaceL^\infty_\loc\at{\Rset}$, where $u_\eps$ and
$u$ are defined in Lemma \ref{Lem:Cont.ExistenceOfWaves} and Remark \ref{Rem:Cont.LimitSol},
respectively.
\end{corollary}
\begin{proof}
Let $\at{\eps_n}_n\subset\Rset$ be any sequence with $\eps_n\to0$, and let
$u_n=u_{\eps_n}\in\calM_{\eps_{n}}$. By weak$\star$ compactness in
$\fspaceL^\infty\at{\Rset}$ we can extract a (not relabelled) subsequence that converges
weakly$\star$ to some limit $u_\infty\in\calM$. Lemma \ref{Lem:Cont.ExistenceOfWaves}
guarantees that the sequence $\at{v_n}_n$ with
\begin{align*}
v_n\at\xi=\frac{u_n\at{\xi+\eps_n}-u_n\at\xi}{\eps_n}
\end{align*}
is bounded in $\fspaceL^2\at{\Rset}$, and passing to a further subsequence we can assume that
$v_n$ converges weakly in $\fspaceL^2\at{\Rset}$ to some limit $v_\infty$. We then have
\begin{align*}
0=\int\limits_{\Rset}v_{n}\at\xi\varphi\at\xi\dint{\xi}+
\int\limits_{\Rset}u_n\at\xi\frac{\varphi\at{\xi}-\varphi\at{\xi-\eps_n}}{\eps_n}\dint{\xi}
\quad\quad\forall\;\varphi\in\fspaceC^\infty_0\at{\Rset},
\end{align*}
and passing to the limit $n\to\infty$ we conclude that $v_\infty$ is the weak derivative of
$u_\infty$. In particular, $u_\infty$ is continuous and satisfies
\begin{align*}
u_\infty\at{\xi}=\int\limits_{0}^{\xi}v_\infty\at{\zeta}\dint{\zeta}=
\lim\limits_{n\to\infty}\int\limits_{0}^{\xi}v_n\at{\zeta}\dint{\zeta}
=\lim\limits_{n\to\infty}u_n\at{\xi},
\end{align*}
where we used that $u_n\at{\xi}=\int_{0}^{\xi}v_n\at{\zeta}\dint{\zeta}+\DO{\eps_n}$ holds by
construction. Therefore, $u_n$ converges to $u_\infty$ in $\fspaceL^\infty_\loc\at{\Rset}$,
so \eqref{Lem:Cont.ExistenceOfWaves.Eqn1} implies that $u_\infty$ and $u$ solves the same
ODE. Moreover, Lemma \ref{Lem:Cont.ExistenceOfWaves} combined with Fatou's Lemma provides
\begin{align*}
\calF\at{u}\leq\liminf\limits_{n\to\infty}\calF\at{u_n}<\infty,
\end{align*}
which implies $u_\at\xi\to\pm1$ as $\xi\to\pm\infty$ and hence $u_\infty=u$. We have now
shown that each sequence $\at{u_{\eps_n}}_n$ has at least a subsequence that converges to
$u$. The uniqueness of $u$ now implies the desired result.
\end{proof}
%
%
%
\section{Numerical simulations}\label{sec:num}
%
We illustrate the analytical results from the previous section by numerical simulations of
 discrete waves. To this end we implemented the explicit Euler-scheme for the
gradient flow of $\calE$ in the Ritz set $\calM_N$ from \eqref{Eqn:RitzSet}. The
corresponding iteration scheme with flow time $\tau>0$ reads
\begin{align*}
u\in\calM_N\mapsto\calI\at{u}\in\calM_N,
\end{align*}
with
\begin{align*}
\calI\at{u}_j=u_j-\tau{}\at{F^\prime\at{u_j}-2\beta\at{u_{j+1}+u_{j-1}-2u_{j}}}
\quad\forall\;\abs{j}\leq{N}.
\end{align*}
\begin{figure}[t!]
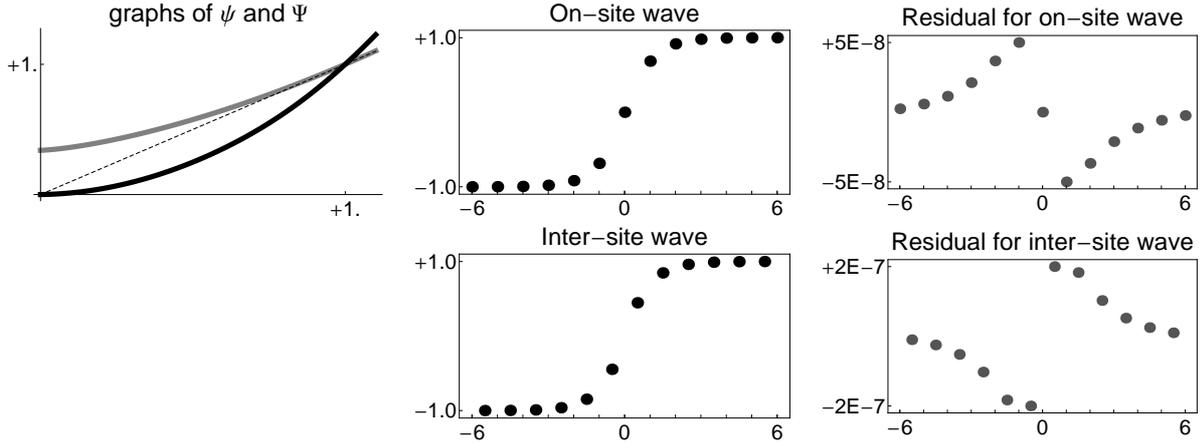
%
\centering{%
\begin{minipage}[c]{0.3\textwidth}%
\includegraphics[width=\textwidth, draft=\figdraft]%
{\figfile{ex_1_graph}}%
\end{minipage}%
\hspace{0.025\textwidth}%
\begin{minipage}[c]{0.3\textwidth}%
\includegraphics[width=\textwidth, draft=\figdraft]%
{\figfile{ex_1_wave_1}}%
\end{minipage}%
\hspace{0.025\textwidth}%
\begin{minipage}[c]{0.3\textwidth}%
\includegraphics[width=\textwidth, draft=\figdraft]%
{\figfile{ex_1_res_1}}%
\end{minipage}%
\\%
\begin{minipage}[c]{0.3\textwidth}%
\quad
\end{minipage}%
\hspace{0.025\textwidth}%
\begin{minipage}[c]{0.3\textwidth}%
\includegraphics[width=\textwidth, draft=\figdraft]%
{\figfile{ex_1_wave_2}}%
\end{minipage}%
\hspace{0.025\textwidth}%
\begin{minipage}[c]{0.3\textwidth}%
\includegraphics[width=\textwidth, draft=\figdraft]%
{\figfile{ex_1_res_2}}%
\end{minipage}%
}%
\caption{%
On-site and inter-site waves for $\beta=0.25$ after $150$ steps with $\tau=0.1$ and $N=6$.
The upper left picture shows the graphs of $\psi$ (Black) and $\Psi$ (Gray),
see \eqref{Eqn:DefLittePsi}.
}%
\label{Fig:NumEx1}%
\end{figure}%
\begin{figure}[t!]
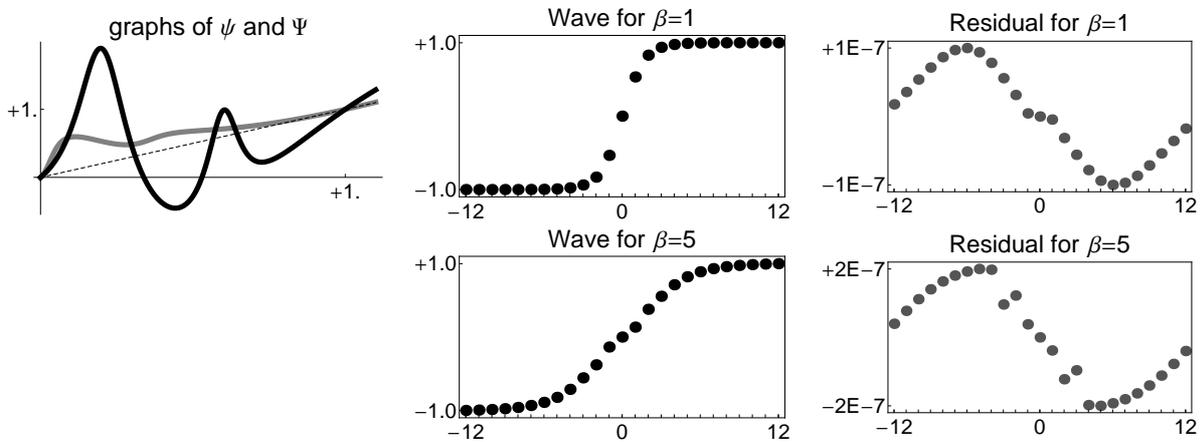
%
\centering{%
\begin{minipage}[c]{0.3\textwidth}%
\includegraphics[width=\textwidth, draft=\figdraft]%
{\figfile{ex_2_graph}}%
\end{minipage}%
\hspace{0.025\textwidth}%
\begin{minipage}[c]{0.3\textwidth}%
\includegraphics[width=\textwidth, draft=\figdraft]%
{\figfile{ex_2_wave_1}}%
\end{minipage}%
\hspace{0.025\textwidth}%
\begin{minipage}[c]{0.3\textwidth}%
\includegraphics[width=\textwidth, draft=\figdraft]%
{\figfile{ex_2_res_1}}%
\end{minipage}%
\\%
\begin{minipage}[c]{0.3\textwidth}%
\quad%
\end{minipage}%
\hspace{0.025\textwidth}%
\begin{minipage}[c]{0.3\textwidth}%
\includegraphics[width=\textwidth, draft=\figdraft]%
{\figfile{ex_2_wave_2}}%
\end{minipage}%
\hspace{0.025\textwidth}%
\begin{minipage}[c]{0.3\textwidth}%
\includegraphics[width=\textwidth, draft=\figdraft]%
{\figfile{ex_2_res_2}}%
\end{minipage}%
}%
\caption{%
On-site waves for $\beta=1$ and $\beta=5$ after $600$ steps with $\tau=0.01$ and $N=12$.
}%
\label{Fig:NumEx2}%
\end{figure}%
The iteration scheme works very well in numerical simulations and converges to a fixed point
of $\calI$ provided that $\tau$ is sufficiently small. It is clear, at least for large $N$,
that each fixed point of $\calI$ satisfies the standing wave equation up to high order, but
it may happen that this fixed point approximates a \emph{local} minimizer of $\calE$.
\par
To
ensure that the scheme converges to a global minimizer we need a good guess for the initial
profile. In our simulations we always start with the shock profile $u_j=\sgn{j}$ for all $j$,
which is the global minimizer for $\beta=0$.
\par
Typical numerical solutions for convex and non-convex $\Psi$ are shown in Figure
\ref{Fig:NumEx1} and \ref{Fig:NumEx2}, respectively. In the upper left picture we also
plotted the graph of
\begin{align}
\label{Eqn:DefLittePsi}%
\psi\at{\eta}=\Psi^\prime\at{\eta^2}\eta,
\end{align}
which is the nonlinearity in the standing wave equation \eqref{Eqn:StandingWave}. Notice that
the first example is prototypical for power laws $\psi\at{\eta}=\tfrac{1}{1+d}\eta^{1+d}$
with $d>1$.
\begin{figure}[t!]
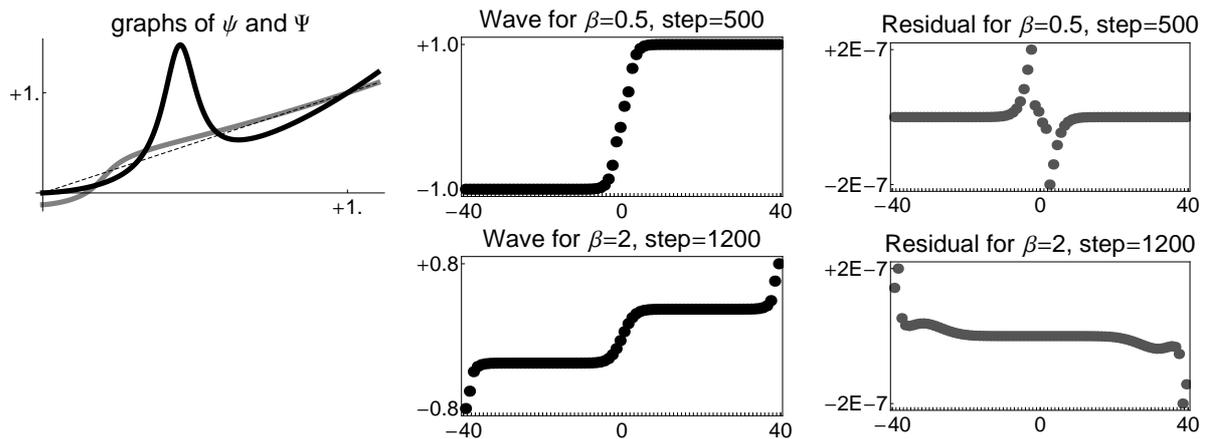
%
\centering{%
\begin{minipage}[c]{0.3\textwidth}%
\includegraphics[width=\textwidth, draft=\figdraft]%
{\figfile{ex_3_graph}}%
\end{minipage}%
\hspace{0.025\textwidth}%
\begin{minipage}[c]{0.3\textwidth}%
\includegraphics[width=\textwidth, draft=\figdraft]%
{\figfile{ex_3_wave_1}}%
\end{minipage}%
\hspace{0.025\textwidth}%
\begin{minipage}[c]{0.3\textwidth}%
\includegraphics[width=\textwidth, draft=\figdraft]%
{\figfile{ex_3_res_1}}%
\end{minipage}%
\\%
\begin{minipage}[c]{0.3\textwidth}%
\quad%
\end{minipage}%
\hspace{0.025\textwidth}%
\begin{minipage}[c]{0.3\textwidth}%
\includegraphics[width=\textwidth, draft=\figdraft]%
{\figfile{ex_3_wave_2}}%
\end{minipage}%
\hspace{0.025\textwidth}%
\begin{minipage}[c]{0.3\textwidth}%
\includegraphics[width=\textwidth, draft=\figdraft]%
{\figfile{ex_3_res_2}}%
\end{minipage}%
}%
\caption{%
Off-site waves for $\beta=.5$ and $\beta=2$ after $1000$ steps with $N=40$ and $\tau=0.05$.
}%
\label{Fig:NumEx3}%
\end{figure}%
\begin{figure}[t!]
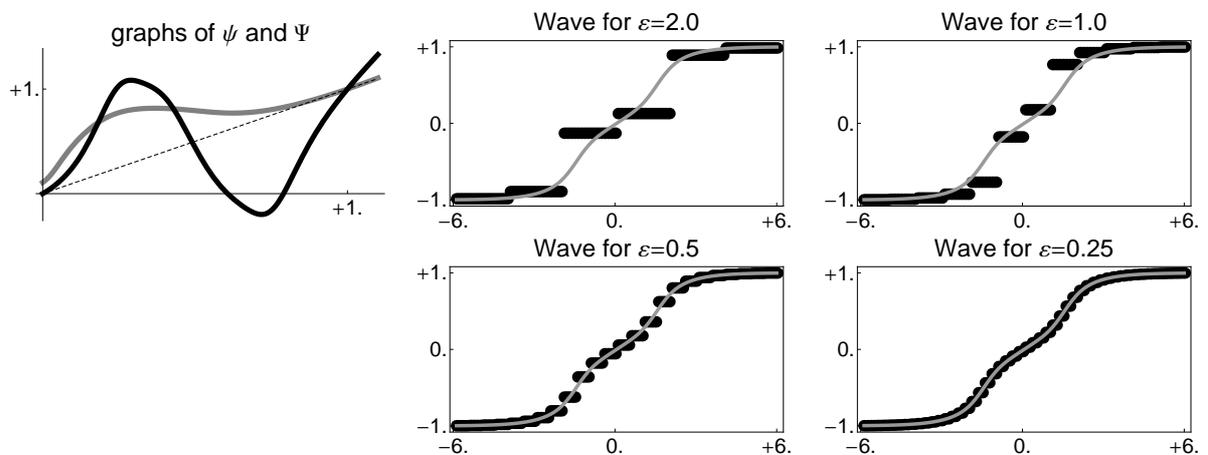
%
\centering{%
\begin{minipage}[c]{0.3\textwidth}%
\includegraphics[width=\textwidth, draft=\figdraft]%
{\figfile{ex_4_graph}}%
\end{minipage}%
\hspace{0.025\textwidth}%
\begin{minipage}[c]{0.3\textwidth}%
\includegraphics[width=\textwidth, draft=\figdraft]%
{\figfile{ex_4_wave_1}}%
\end{minipage}%
\hspace{0.025\textwidth}%
\begin{minipage}[c]{0.3\textwidth}%
\includegraphics[width=\textwidth, draft=\figdraft]%
{\figfile{ex_4_wave_2}}%
\end{minipage}%
\\%
\begin{minipage}[c]{0.3\textwidth}%
\quad%
\end{minipage}%
\hspace{0.025\textwidth}%
\begin{minipage}[c]{0.3\textwidth}%
\includegraphics[width=\textwidth, draft=\figdraft]%
{\figfile{ex_4_wave_3}}%
\end{minipage}%
\hspace{0.025\textwidth}%
\begin{minipage}[c]{0.3\textwidth}%
\includegraphics[width=\textwidth, draft=\figdraft]%
{\figfile{ex_4_wave_4}}%
\end{minipage}%
}%
\caption{%
Continuum limit of inter-site waves with $\beta=1$: Discrete waves (Black) from $\calM_\eps$
for several values of $\eps$ along with the solution (Gray) of the limit problem
\eqref{Eqn:LimitProblem}.
}%
\label{Fig:NumEx4}%
\end{figure}%
\bigpar
Figures \ref{Fig:NumEx3} illustrates what happens if $\Psi$ violates
\eqref{Intro:Result1.Eqn1}. In this example, $F$ is not positive in $\oointerval{-1}{1}$ but
there exists $0<\eta_\ast<1$ such that
$0>F\at{\eta_\ast}=F\at{-\eta_\ast}=\min_{-1\leq\eta\leq1}{F\at{\eta}}$. Consequently,
$\calE|_\calM$ is unbounded from below and global minimizers $u\in\calM$ cannot exit. In
fact, setting $u_j=\eta_\ast\,\sgn{j}$ for all $\abs{j}\leq{N}$ and $u_j=\sgn{j}$ for
$\abs{j}>N$ we find $\calE\at{u}=N{}F\at{\eta_\ast}+\DO{1}\to-\infty$ as $N\to\infty$. We now
interpret the numerical results in Figure \ref{Fig:NumEx3}. For small $\beta$ (and large $N$)
the iteration scheme converges to a fixed point of $\calI$ that satisfies the standing wave
equation up to high order and has a sharp transition from $-1$ to $+1$. We therefore
conjecture that standing waves with $u_{\pm\infty}=\pm1$ still exist and correspond to local
minimizers of $\calE|_\calM$. For larger values of $\beta$, however, the solutions
$u_N\in\calM_N$ exhibit plateaus at heights $\pm\eta_\ast$ and converge for $N\to\infty$ to
some limit profile $u$ with asymptotic states $u_{\pm\infty}=\pm\eta_{\ast}$. Notice that $u$
solves the standing wave equation \eqref{Eqn:StandingWave} as $F^\prime\at{\eta_\ast}=0$
implies $\Psi^\prime\at{\eta^2_\ast}=1$.
\par
Finally, Figure \ref{Fig:NumEx4} concerns  the continuum limit of standing waves, see
\S\ref{sec:Waves.Limit}. It shows the piecewise constants minimizers $u_\eps\in\calM_\eps$
for different values of $\eps$ in the interval $\abs{\xi}\leq6$, and illustrates that
$u_\eps$ converges for $\eps\to0$ to the unique solution of \eqref{Eqn:LimitProblem}.
%
%
%
%
%
\providecommand{\bysame}{\leavevmode\hbox to3em{\hrulefill}\thinspace}
\providecommand{\MR}{\relax\ifhmode\unskip\space\fi MR }
\providecommand{\MRhref}[2]{%
  \href{http://www.ams.org/mathscinet-getitem?mr=#1}{#2}
}
\providecommand{\href}[2]{#2}

\end{document}